\documentclass{article}

\usepackage{arxiv}

\usepackage[utf8]{inputenc} 
\usepackage[T1]{fontenc}    
\usepackage{hyperref}       
\usepackage{url}            
\usepackage{booktabs}       
\usepackage{amsfonts}       
\usepackage{nicefrac}       
\usepackage{microtype}      
\usepackage{lipsum}

\usepackage{amsthm}
\usepackage{amsmath}

\theoremstyle{plain}
 \newtheorem{thm}{Theorem}[section]
 \newtheorem{lem}{Lemma}[section]
 
 \newtheorem{prop}{Proposition}[section]
\theoremstyle{definition}
  \newtheorem{defn}{Definition}[section]
  
  \newtheorem{ass}{Assumption}[section]
\theoremstyle{remark}
  \newtheorem{rem}{Remark}

\newcommand{\Cal}[1]{\ensuremath{\mathcal{#1}}}
\newcommand{\g}[1]{\mbox{\boldmath ${#1}$}}

\newcommand{\R}{\mathbb{R}}

\newcommand{\p}{\partial}
\newcommand{\bs}{\mathbb{S}}
\newcommand{\ve}{\varepsilon}

\title{Inverse $N$-body scattering 
with the time-dependent Hartree-Fock approximation}

\author{
 Michiyuki~Watanabe \\
  Faculty of Education\\
  Niigata University\\
  Niigata, Japan \\
  \texttt{mwatanab@ed.niigata-u.ac.jp} \\
}

\begin{document}
\maketitle

\begin{abstract}
We consider an inverse $N$-body scattering problem of determining 
two potentials---an external potential 
acting on all particles and a pair interaction potential---from the scattering 
particles. 
This paper finds that the time-dependent Hartree-Fock approximation for a 
three-dimensional inverse $N$-body scattering in quantum mechanics 
enables us to recover the two potentials from the scattering states with high-velocity initial states. 
The main ingredient of mathematical analysis in this paper is based on the asymptotic analysis of
the scattering operator defined in terms of a scattering solution to the Hartree-Fock equation at high energies. We show that the leading part of 
the asymptotic expansion of the scattering operator 
uniquely reconstructs the Fourier transform of the pair interaction, 
and the second term of the expansion uniquely reconstructs the $X$-ray transform of the external potential.
\end{abstract}

\keywords{Hartree-Fock approximation 
\and Inverse scattering problems 
\and Non-linear Schr\"{o}dinger equations
\and Quantum $N$-body scattering
}

\section{Introduction}
\subsection{Problem and result}
Consider the quantum $N$-body systems of identical particles interacting pairwise 
by the two-body potential under an external potential acting on all particles. 
A typical example is $N$ electrons in an atom with proton number $Z$ at the nucleus. In that case, the external potential is the 
nucleus-electron attraction, and the two-body potential is the 
electron to electron repulsion. 
Inverse $N$-body scattering problems ask to determine the interaction potential 
and the external potential from the scattering states of particles.
Such inverse problems have been extensively studied for $N$-body Schr\"{o}dinger 
equations with no external potentials (Enss and Weder \cite{Enss-Weder1995}; 
Novikov \cite{Novikov}; Wang \cite{Wang, Wang1996}; Vasy \cite{Vasy}; Uhlmann and Vasy \cite{Uhlmann-Vasy, Uhlmann-Vasy2003, Uhlmann-Vasy2004}). 
The inverse scattering for the $N$-body Schr\"{o}dinger equation in an external 
constant electric field was investigated by Valencia and Weder \cite{Valencia-Weder2012}.

Differently, Lemm and Uhlig \cite{Lemm-Uhlig2000} have investigated an 
inverse $N$-body problem by using 
Bayesian approach with the Hartree-Fock approximation. 
They gave a computationally feasible 
method of reconstructing an interaction potential from data by 
solutions to a stationary Hartree-Fock equation. 
Their work indicates that the Hartree-Fock approximation is 
also extremely useful as a way to investigate the 
inverse $N$-body problems. 

The above mentioned works have focused only on recovering interactions. 
Since $N$-body systems is generally described by a non-relativistic Hamiltonian consisting  of a one-body term with the kinetic energy and an external potential, and a two-body interaction term, 
the inverse problems of determining both the interaction potential and 
the external potential should be also investigated. 
However, little has been reported on the determination  both
the interaction potential and the external potential in the 
quantum $N$-body systems.   

In this paper, we find that the time-dependent Hartree-Fock approximation for the inverse $N$-body scattering in quantum mechanics enables us to recover two potentials---an external potential 
acting on all particles and a pair interaction potential---from the scattering states with high-velocity initial states. 
This paper also propose a new reconstruction procedure of recovering the two potentials. 

Let us formulate our inverse problem and state our main result. 
We first recall that the $n$-dimensional $N$-body Schr\"{o}dinger equation has the form: 
\begin{align*}
   & i\frac{\p}{\p t}\Psi(t)  = \widetilde{H}_N \Psi (t), \\
   & \widetilde{H}_N = \sum_{j=1}^N \left[ \frac{1}{2}\left( -i \nabla_{\mathbf{x}_j} \right)^2  +V_{ext}(\mathbf{x}_j)\right]
    + \sum_{j<k}^N V_{int}(\mathbf{x}_j - \mathbf{x}_k),
\end{align*}
where $i=\sqrt{-1}$, $\mathbf{x}_j\in \R^n$,  
$V_{ext}(\mathbf{x}_j)$ is an external potential and 
$V_{int}(\mathbf{x}_j)$ is an interaction potential with  
$V_{int}(\mathbf{x}_j)=V_{int}(-\mathbf{x}_j)$. 
The Hartree-Fock approximation is known as the simplest one-body  approximation. 
Writing the $N$-body wave function 
$\Psi(t)=\Psi (t, \mathbf{x}_1, \cdots , \mathbf{x}_N)$ 
with the Slater determinant 
\[
   \Psi (t, \mathbf{x}_1, \cdots , \mathbf{x}_N) 
    = (N!)^{-1/2} {\rm det} \left( 
     u_j (t, \mathbf{x}_k)
    \right)_{1\le j, k \le N}
\]
yields the one-body Schr\"{o}dinger equation: 
\begin{align}\label{eqn:1-1}
  i \frac{\p u_j }{\p t} &= H(u_k) u_j, \\
  H(u_k)u_j &= \left[ H_0  + V_{ext} + Q_H(x,\g{u}) \right] u_j  
   + \int_{\R^n} Q_F (x,y, \g{u}) u_j (t, y) \, dy
   \qquad \text{for $1\le j \le N$}, \notag
\end{align}
where $H_0= -\dfrac{1}{2} \Delta= -\dfrac{1}{2}\sum_{j=1}^n \frac{\p^2}{\p x_j^2}$ 
and $\g{u}=\g{u}(t,x)=( u_j(t,x) )_{1\le j\le N}$ is an unknown function in 
$(t,x) \in \R \times \R^n$, and 
 \begin{align*}
     Q_H (x, \g{u}) & =  \int_{\R^n} 
     V_{int}(x-y) \sum_{\substack{k=1 \\ k \not=j}}^N 
      | u_{k}(t, y) |^2 \, dy, \\
     &=
     V_{int}* \sum_{\substack{k=1 \\k\not=j}}^N | u_k (t, \cdot)|^2, \\
   Q_F(x,y, \g{u}) &= - V_{int}(x-y) \sum_{\substack{k=1 \\k\not=j}}^N \overline{u_k} (t, y) u_k (t, x). 
 \end{align*}
The non-linear Schr\"{o}dinger equation \eqref{eqn:1-1} 
we study in this paper 
is called the Hartree-Fock equation (HF equation). 
The terms  
$Q_H( x,\g{u} ) u_j (t,x)$  and  $\int Q_F(x,y,\g{u}) u_j (t,y) dy$ 
are called the Hartree term and the Fock term, respectively.  

Next, we introduce some notations and assumptions on the 
potentials. 
Let $W^{k,p}(\R^n)$ be the usual Sobolev space in $L^p(\R^n)$. 
We abbreviate $W^{k,2}(\R^n)$ as $H^k(\R^n)$. 
The weighted $L^2$-space is denoted as 
 \[
   L^{2,s}(\R^n) =\left\{ u(x)\, : \, (1+|x|^2)^{s/2}u(x) \in L^2(\R^n) , \, s\in \R \right\}.
 \]
Let $C_0^{\infty}(\R^n)$ be the set of compactly supported smooth functions and 
$\Cal{S}(\R^n)$ be the set of rapidly decreasing functions on $\R^n$. 
The Fourier transform is denoted as 
 \[
    \left( \Cal{F}u \right)(\xi) = \widehat{u}(\xi) = 
    \frac{1}{(2\pi)^{n/2}}\int_{\R^n} e^{-ix\cdot \xi} u(x) \, dx.
 \]
We define a function space $\Cal{S}_0(\R^n)$ as 
 \[
   \Cal{S}_0 (\R^n) = \left\{ f\in \Cal{S}(\R^n)\, ; \, \widehat{f} \in C_0^{\infty}(\R^n)\right\}.
 \]
The multiplication operator with a fixed function $V(x)$ is denoted as $V$. 
The unitary group of the self-adjoint operator $H_0$ with a domain 
$H^1(\R^n)$ is denoted as $U_0(t)$ or $e^{-itH_0}$. Then, 
solutions of the free Schr\"{o}dinger equation 
$i\p_t v = H_0 v$ with initial data $v(0)=f$ is written as 
$v(t)=U_0(t)f = e^{-itH_0}f$. 
Consider solutions to the equation \eqref{eqn:1-1} with 
$u_j(t) \longrightarrow U_0 (t) f_j^{\pm}$ 
as $t\to \pm \infty$ in some function space.  
We term the solutions scattering solution and $f_j^{\pm}$ scattering states. 
The scattering operator $S$ assigns the free state $U_0(t) f_j^-$ at 
$t=-\infty$ to the free state $U_0(t)f_j^+$ at $t=+\infty$, or equivalently 
$S: f_j^- \to f_j^+$. 
Our goal is to recover the external potential $V_{ext}(x)$ and 
the interaction potential $V_{int}(x)$ from the scattering operator $S$.

Although we consider the three-dimensional inverse problem,  
to make it easy to explain a proof of 
our theorem, we will denote the spatial 
 dimension by $n$ throughout this paper.

Let potentials satisfy the following conditions. 
\begin{ass}\label{ass:interaction-2}
 Let $n=3$. 
 We assume that the real-valued function $V_{int}(x)$ has the 
 following conditions: 
 \begin{enumerate}
     \item $V_{int}(x) \ge 0$ and 
 \begin{equation*}
     | V_{int}(x)| \le C | x |^{-2}, \qquad \text{or} \quad 
     V_{int}\in L^{n/2}(\R^n).
 \end{equation*}
 \item $\nabla V_{int} \in L^{n/2}(\R^n)$. 
 \item $V_{int}\in L^q (\R^n) + L^{\infty}(\R^n)$ with 
 $1\le q$. 
 \item   $x\cdot \nabla V_{int}\in L^{\delta} (\R^n)+ L^{\infty}(\R^n)$ with 
 $1\le \delta$. 
 \item $V_{int}(-x)=V_{int}(x)$. 
 \item $|x|^2 V_{int}(x)$ is a non-increasing function of $|x|$.
 \item $\sup_{x\in\R^n}(1+|x|)^{1+s}|V_{int}(x)|<\infty$ for $s>n/2$.
 \end{enumerate}
\end{ass}

\begin{ass}\label{ass:external}
 Let $n=3$. 
 We assume that the real-valued function $V_{ext}(x)$ has the following 
 conditions: 
 \begin{enumerate}
     \item $V_{ext}(x)\ge 0$.
     \item $V_{ext}(x)$ is a homogeneous function of degree $-\gamma$: 
     \[ V_{ext}(\alpha x) = \alpha^{-\gamma}V_{ext}(x), \qquad 
     \text{for $\alpha >0$ and  $\gamma \ge 1$}.
     \]
 \item $|x|^2 V_{ext}(x)$ is a non-increasing function of $|x|$.
 \item Zero is not an eigenvalue of the Schr\"{o}dinger operator 
 $H=H_0 + V_{ext}$.
     \item $\nabla V_{ext}\in L^{\infty}(\R^n)$ and 
           $\Delta V_{ext} \in L^n(\R^n)$. 
     \item $V_{ext}\in L^p (\R^n) + L^{\infty}(\R^n)$ with 
 $1\le p$.
 \item  
 $x\cdot \nabla V_{ext}\in L^{\delta} (\R^n) + L^{\infty}(\R^n)$ with 
 $1\le \beta$. 
     \item Let $\ell \ge 0$ be an arbitrary fixed integer.  
     For $\delta > 3n/2 +1$, $p_0 > n/2$ and multi-indices $\alpha$ with 
     $|\alpha|\le \ell$, 
     \begin{equation*}
         \sup_{x\in\R^n}(1+|x|^2)^{\delta/2} \left( \int_{|x-y|\le 1} \left| D^{\alpha} V_{ext}(y) \right|^{p_0} \, dy \right)^{1/p_0} <\infty.
     \end{equation*}
 Here we have denoted the integer part of $x$ by $[x]$ and 
 $D^{\alpha}=D_1^{\alpha_1}\cdots D_n^{\alpha_n}$, $D_j=-i\frac{\p}{\p x_j}$.
 \end{enumerate}
\end{ass}

\begin{rem} 
A proof of the unique existence theorem on the scattering solution requires 
the $L^p$-decay of solutions of the Cauchy problems 
for time-dependent Schr\"{o}dinger equations: $i\frac{\p u}{\p t}= Hu$. 
The condition $V_{ext}(x)\ge 0$ causes an absence of zero resonance for the 
Schr\"{o}dinger operator $H$. Then, under the conditions 1, 4 and 8 in Assumption \ref{ass:external},  
the $W^{k, p}(\R^n)$-continuity of wave operators for the Schr\"{o}dinger operator $H$ for any $k=0,1,\cdots, \ell$ and $1\le p \le \infty$ follows  
from the result developed by Yajima \cite{Yajima1995-2}, which implies  
 the $L^p$-decay of the solutions. 
\end{rem}

\begin{rem} The Assumption \ref{ass:external} is rather complicated condition. 
We therefore give another conditions on $V_{ext}(x)$ 
simpler than Assumption \ref{ass:external}.
\begin{ass}\label{ass:external-2}
 Let $n=3$. 
 We assume that the real-valued function $V_{ext}(x)$ has the following 
 conditions: 
 \begin{enumerate}
     \item $V_{ext}(x)\ge 0$.
     \item $V_{ext}(x)$ is a continuously differentiable function and 
     a homogeneous function of degree $-\gamma$ with $\gamma \ge 2$.
 \item Zero is not an eigenvalue of the Schr\"{o}dinger operator 
 $H=H_0 + V_{ext}$.
     \item For $|\alpha | \le 2$ and $\kappa >(3n)/2 +3$, 
      \[ 
    |D^{\alpha} V_{ext}(x) | \le \frac{C}{(1+|x|)^{\kappa}}.
 \]
 \end{enumerate}
\end{ass}
The Assumption \ref{ass:external} includes the Assumption \ref{ass:external-2}. 
Indeed, letting $r=|x|$, we have 
\[
   V+\frac{1}{2} x\cdot \nabla V_{ext} = V+\frac{1}{2}r \p_r V_{ext}
   = \frac{1}{2r} \p_r (r^2 V_{ext}).
\]
This identity means that the condition $V+\frac{1}{2}x\cdot \nabla V_{ext}\le 0$ replaces the condition 3 in Assumption \ref{ass:external}. 
Recall the Euler's homogeneous function theorem:  
if the function $V(x)$ on $\R^n$ is  continuously differentiable 
function, then $V(x)$ is a positively homogeneous of degree $\gamma$ 
if and only if  $V(x)$ satisfies $x\cdot \nabla V(x)=\gamma V(x)$. 
Then, in view of the Euler's homogeneous function theorem, 
the condition 2 in Assumption \ref{ass:external} gives $x\cdot V_{ext}=-\gamma V_{ext}\le 0$, which implies 
\[
   V_{ext}+\frac{1}{2} x \cdot \nabla V_{ext}=V_{ext}-\frac{\gamma}{2}V_{ext} 
   \le 0
\]
for $\gamma \ge 2$. 
Direct computations show that the function $V_{ext}(x)$ satisfies 
conditions 5-8 in Assumption \ref{ass:external}. 
\end{rem}
As it turns out in Section \ref{sec:2}, under the Assumption \ref{ass:interaction-2} and the Assumption \ref{ass:external}, 
there exists a unique scattering solution 
$u_j(t,x)$ of \eqref{eqn:1-1} with a condition 
$u_j(t,x)\to e^{-itH_0}\varphi_{j}$ as 
$t\to -\infty$ in $L^2(\R^n)$ 
for any $\varphi_{j}\in \Cal{S}_0$ , $j=1,\cdots , N$ sufficiently close to 
zero function. 
Put  
\begin{align}
    \psi_j(x)= (S\g{\varphi})_{j} (x) 
     & := \varphi_{j} (x) + \frac{1}{i} \int_{\R}
     e^{itH_0} P_j(x,\g{u}) \, dt, \qquad \text{$ j= 1, 2, \cdots, N$} 
     \label{eqn:1-2}\\
     P_j(x,\g{u}) 
     &=
     \left( Q_{H}(x,\g{u}) + V_{ext}(x) \right) u_j(t,x) 
     + \int_{\R^n}Q_{F}(x,y,\g{u}) u_j (t,y)\, dy, \notag
\end{align}
where $\g{u} (t,x)$ is the scattering solution to \eqref{eqn:1-1}. 
Then it will be shown that 
$u_j(x,t) \to e^{-itH_0}\psi_j(x)$ as $t\to +\infty$ 
in $L^2(\R^n)$. 
Therefore, the operator $S$ defined as \eqref{eqn:1-2} represents a scattering operator for the HF equation 
\eqref{eqn:1-1}.

The inverse problem considered in this paper is to determine the interaction and the external potentials from the scattering operator defined in terms of the scattering solution to the Hartree-Fock equation \eqref{eqn:1-1}. 
Our main result is  
\begin{thm}\label{thm:1-1}
Let $n=3$. 
Assume that $V_{int}(x)$ and $V_{ext}(x)$ satisfy
Assumption \ref{ass:interaction-2} and Assumption \ref{ass:external}, 
respectively. 
Then the potentials $V_{int}$, $V_{ext}$ are 
uniquely determined by $S$.
\end{thm}

We remark that our proof gives an explicit way to reconstruct the interaction and the external  potentials from the asymptotic behavior of the function 
$<(S-I)\g{\Psi}_v, \g{\Psi}_v>_{L^2}$ at $|v|\to \infty $, where 
$\g{\Psi}_v(x)=e^{iv\cdot x}\g{\varphi}(x)$ and $<\, ,\, >_{L^2}$ is the inner product in $L^2(\R^n)$.

\subsection{Methods}
Because the high velocity limit (HVL) of the scattering operator (Enss and Weder \cite{Enss-Weder1995}) makes it possible to recover the Schr\"{o}dinger operator with 
the potentials, it has become an important tools for studying the inverse scattering problems for time-dependent Schr\"{o}dinger equations (Weder \cite{Weder1997-2}; Valencia and Weder \cite{Valencia-Weder2012}, and references therein; Adachi and Maehara \cite{Adachi-Maehara2007}; Adachi, Fujiwara and Ishida \cite{Adachi-Fujiwara-Ishida2013}; Adachi et al. \cite{Adachi.et2011}; Ishida \cite{Ishida2019}). 
According to recent researches for inverse nonlinear scattering (\cite{Watanabe2018} and \cite{Watanabe2019}), the method of the HVL also make it possible to recover nonlinearities. 
On the other hand, the small amplitude limit (SAL) of the scattering operator (Weder \cite{Weder1997,Weder1999, Weder1999_2, Weder2000_1, Weder2000_2, Weder2000_3, Weder2001_1, Weder2001_2, Weder2002, Weder2003}) make it possible to recover both the potential and nonlinearities. 
This method of the SAL, however, fails to  reconstruct the general interaction potential in the HF equation \eqref{eqn:1-1}. 
Details of the method of SAL are described briefly as follows. 

Because the HF equation \eqref{eqn:1-1} is a non-linear equation,  
our inverse problem is a non-linear inverse problem of recovering 
the linear part---zero-order coefficient--- and non-linear part. 
Such non-linear inverse scattering problem have been extensively studied by 
Weder. 
It was proved that the SAL of the scattering operator 
uniquely determines coefficients---the coefficient of the linear part and 
coefficients of the power type non-linearity. 
In other words, this method of the SAL is an asymptotic analysis 
of the scattering operator $S(\varepsilon \varphi)$ as $\varepsilon \to 0$. 
In particular, Weder \cite{Weder1997} showed that the Fr\'{e}chet derivative of the scattering operator $S$ uniquely determines the linear scattering operator for the Schr\"{o}dinger operator $H=H_0+V_{ext}$. 
Thus, the non-linear inverse scattering problem of recovering the potential $V_{ext}$ is reduce to the problem of recovering the 
Schr\"{o}dinger operator $H$ from the linear scattering operator. 

This method of Weder was applied to inverse problems for Hartree equations 
(\cite{Watanabe0,Watanabe2007-1}). 
Here, we briefly review the method to 
recover the coefficient function $V_{ext}$ of the linear term 
in the case of $N=2$ and $u_1=u_2$ to the equation \eqref{eqn:1-1}.
Following Weder \cite{Weder1997}, the scattering operator $S$ is defined in terms of wave operators $W_{\pm}=\lim_{t\to \pm\infty}e^{itH} e^{-itH_0}$ for the Schr\"{o}dinger operator $H=H_0+V_{ext}$: 
\begin{align*}
    S & = W_+^* S_N W_-, \\
    (S_N \g{\varphi})_j(x)& = 
    \varphi_j(x) +\frac{1}{i} \int_{\R} e^{itH}N_j(x, \g{u})\, dt,  \\
    N_j(x, \g{u}) &= Q_H (x, \g{u})u_j(t,x) + 
     \int_{\R^n} Q_F (x,y, \g{u})u_j(t,y)\, dy,   
\end{align*}
and the scattering operator $S$ has expansion 
\[ 
   S(\varepsilon\g{\varphi}) = \varepsilon S_{V_{ext}} \g{\varphi} + O(\varepsilon^3)
\]
as $\varepsilon \to 0$ in $H^1(\R^n)$, where 
$S_{V_{ext}}$ denotes the scattering operator for 
the Schr\"{o}dinger operator $H$.
We shall term this expansion ``$\varepsilon$-expansion''. 
This $\ve$-expansion indicates that the scattering operator 
$S$ uniquely determines 
the scattering operator $S_{V_{ext}}$. 
As is well-known (see, e.g., \cite{Enss-Weder1995}), 
the operator $S_{V_{ext}}$ uniquely determines the external potential $V_{ext}$. Then we can construct wave operators $W_+=\lim_{t\to \infty}e^{itH} e^{-itH_0}$ and $W^*_-$. Defining $S_F$ as 
$S_F = W_+ S W_-^*$, the small amplitude limit of the function 
$\frac{1}{\varepsilon^3}(S_F-I)(\varepsilon\g{\varphi})$ uniquely determines the interaction potential of the form $V_{int}(x)= \lambda |x|^{-\sigma}$ (see \cite{Watanabe2007}).

This method fails to reconstruct the general interaction potential due to the difficulty of analysis to the operator $e^{itH}$. 
In order to overcome this difficulty, we give another representation of the scattering operator \eqref{eqn:1-2} with no operator $H$. 
Recently, the general interaction potential  
is successfully reconstructed by means of the method of the HVL of the scattering operator in the case of 
the Hartree-Fock equation \eqref{eqn:1-1} with no external potential $V_{ext}$ (\cite{Watanabe2019}). 
Our representation \eqref{eqn:1-2} permits the high-velocity analysis of 
the scattering operator to reconstruct both the interaction potential and the external potential.
More precisely, analyzing the asymptotic expansion of the function 
$<(S-I)\g{\Psi}_v, \g{\Psi}_v>_{L^2}$ 
at $|v|\to \infty$ ($v$-expansion)
shows that the leading term of the expansion  uniquely determines the Fourier transform of the interaction potential $V_{int}$ and the second term of the expansion uniquely determines the 
$X$-ray transform of the external potential $V_{ext}$.  
This method of the $v$-expansion dose not require construction of 
wave operators $W_{\pm}$ to reconstruct the potentials. 
Thus the method of the $v$-expansion is simpler and less complicated 
than the method of the $\ve$-expansion. 

In order to prove that the operator defined as \eqref{eqn:1-2} is a 
scattering operator to the equation \eqref{eqn:1-1}, 
we require a time-decay $L^{\infty}$-estimate on the solution to the Cauchy problem of the equation \eqref{eqn:1-1}. 
Such estimate on the solution to the Hartree type equation 
has been extensively investigated for the case where the potential 
has the form $\lambda |x|^{-\gamma}$ for some constants $\lambda$ 
and $\gamma$ (see, e.g., Wada \cite{Wada}; Hayashi and Naumkin \cite{Hayashi-Naumkin98}; Hayashi and Ozawa \cite{Hayashi-Ozawa1988}). 
Few researchers have addressed the problem of the time-decay 
$L^{\infty}$-estimate on the solution to the Hartree-Fock equation 
with general interaction and external potentials. 
Making use of the pseudo-conformal conservation law and 
the Gagliardo-Nirenberg inequality, with the help of assumptions 
on the potentials, we give a $L^{\infty}$-estimate on the solution 
to the equation \eqref{eqn:1-1}.

This paper is organized as follows: 
Section 2 proves that the operator $S$ defined as \eqref{eqn:1-2}
 is  the scattering operator for the Hartree-Fock equation 
 \eqref{eqn:1-1}, after proving a time-decay 
 $L^{\infty}$-estimate on  the solution to the Cauchy problem 
 for the equation \eqref{eqn:1-1}. 
We give the asymptotic expansion of the scattering operator 
with the high-velocity initial states in Section 3. 
Section 4 is devoted to reconstructions of the interaction potential 
and the external potential.

\section{Representation of 
the scattering operator}
\label{sec:2}
This section shows that the scattering operator for the Hartree-Fock equation \eqref{eqn:1-1} 
has a representation \eqref{eqn:1-2}. 
We first recall that the unique existence of the scattering solution of Hartree equations with the external potential are studied in \cite{Watanabe0}. This result can be easily applicable to the Hartree-Fock equation \eqref{eqn:1-1}. Assume that potentials $V_{ext}$ and $V_{int}$ satisfy Assumption \ref{ass:interaction-2} and Assumption \ref{ass:external}, respectively. 
Then there exists $\varepsilon_0 >0$ such that 
the equation \eqref{eqn:1-1} satisfying the condition 
$u_j(t)\to e^{-itH_0}\varphi_j$ as $t\to -\infty$ in $L^2(\R^n)$ for 
$\varphi_j \in L^2(\R^n)$ with $\| \varphi_j \|_{L^2}\le \varepsilon_0$, $j=1,2, \cdots, N,$ 
has a unique solution  
\[
  u_j \in \Cal{W} = L^3(\R : L^q(\R^n)) \cap 
  L^{\infty}(\R : L^2(\R^n)), 
  \qquad \frac{1}{q}=\frac{1}{2}-\frac{2}{3n}. 
\]
Moreover, there exists a unique $\psi_j\in L^2(\R^n)$ such that $u_j(t)\to e^{-itH_0}\psi_j$ as $t\to \infty$ in $L^2(\R^n)$.

In what follows, 
because we are interested in the inverse scattering problem,  we consider the scattering for high-velocity initial states.  Let $\varphi\in \Cal{S}_0(\R^n)$. 
Then the function $\Phi_v(x)=e^{iv\cdot x}\varphi(x)$ has a compact velocity support in the momentum space around $v$, due to the fact that $\widehat{\Phi}_v (\xi)=\widehat{\varphi}(\xi-v)$.

The main result of this section is 

\begin{thm}\label{thm:2-1}
Let $n=3$. Assume that potentials $V_{int}$ and $V_{ext}$ satisfy Assumption \ref{ass:interaction-2} and Assumption \ref{ass:external}, respectively. Let $u_j$, $j=1,2, \cdots, N$, be the scattering solutions to \eqref{eqn:1-1} with 
initial scattering states $\varphi_j\in \Cal{S}_0$. Put 
\begin{align}
     \psi_j=(S\g{\varphi})_j (x) 
     & = \varphi_j (x) + \frac{1}{i} \int_{\R}
     e^{itH_0} P_j(x,\g{u}) \, dt, \label{eqn:2-1}\\
     P_j(x,\g{u}) 
     &=
     \left( Q_{H}(x,\g{u}) + V_{ext}(x) \right) u_j(t,x) 
     + \int_{\R^n}Q_{F}(x,y,\g{u}) u_j (t,y)\, dy. \notag
\end{align}
Then we have 
\begin{equation*}
    \| u_j(t) - e^{-itH_0}\psi_j \|_{L^2} \longrightarrow 0 \qquad \text{as $t\to \infty$}.
\end{equation*}
\end{thm}

In Subsection \ref{subsec:2-1}, we prepare some lemmas to prove Theorem \ref{thm:2-1}.  
Subsection \ref{subsec:2-2} is devoted to state a $L^{\infty}$ estimate on the solution 
to the equation \eqref{eqn:1-1} and its proof. 
Theorem \ref{thm:2-1} is proved in Subsection \ref{subsec:2-3}.

\subsection{Preliminary lemmas}\label{subsec:2-1}

\begin{lem}[Gagliardo-Nirenberg inequality]\label{lem:Gagliardo}
 Let $q,r$ be any number satisfying $1\le q,r \le \infty$ and let $j,m$ be any integers 
 satisfying $0\le j <m$. Then for any $u\in W^{m,r}(\R^n) \cap L^q(\R^n)$, we have 
  \begin{equation}\label{eqn:Gagliardo}
     \sum_{|\alpha |=j} \left\| D^{\alpha} u \right\|_{L^p} \le 
     M \sum_{|\beta |=m} \left\| D^{\beta} u \right\|_{L^r}^a 
     \left\| u \right\|_{L^q}^{1-a}, 
 \end{equation}
 where $1/p= j/m + a(1/r-m/n)+(1-a)/q$ for all $a\in[j/m, 1]$ with the following exception: 
 if $m-j-(n/r)$ is a non-negative integer, then \eqref{eqn:Gagliardo} is asserted for 
 $a=j/m$, and where $M$ is a positive constant depending only on 
 $n,m,j,q,r,a$.
\end{lem}
The proof of Lemma \ref{lem:Gagliardo} will be found in Friedman \cite{Friedman1969}. 

\begin{lem}\label{lem:2-2}
  Let $1/q=1/2-2/(3n)$. 
  Assume that the potential $V_{int}$  satisfies Assumption \ref{ass:interaction-2}.
  Then for any $u_j\in L^{q}(\R^n)$, $j=1,2,\cdots, N$, we have 
  \begin{equation*}
      \| Q_H(\cdot, \g{u})u_j \|_{L^2}+ 
      \left\| \int_{\R^n} Q_F(\cdot,y, \g{u})u_j(y,t)\, dy \right\|_{L^2} 
      \le C \| \g{u} \|^3_{L^{q}},
  \end{equation*}
  where $C$ is a positive constant. 
\end{lem}

\begin{proof}
 Following Mochizuki \cite[Lemma 4.6.]{Mochizuki}, we obtain  
  \begin{align*}
     \| Q_H(\cdot, \g{u})u_j \|_{L^2} 
     & \le C \sum_{k=1, k\not=j}^N \|u_k\|_{L^{2a}} \| u_k\|_{L^{2b}} \|u_j\|_{L^{2h}}, \\
     \left\| \int_{\R^n} Q_F (\cdot, y,\g{u}) u_j(y,t)\, dy \right\|_{L^2} 
     & \le C \sum_{k=1, k\not=j}^N \left\|  (V_{int}* u_j \overline{u_k})u_k \right\|_{L^2} \\
     & \le \sum_{k=1, k\not=j}^N \|u_j\|_{L^{2a}} \| \overline{u_k} \|_{L^{2b}} 
     \|u_k\|_{L^{2h}}, \\
 \end{align*}
where the positive constants $a, b, h$ satisfy 
$2a=2b=2h=q$, 
due to the H\"{o}lder's inequality,  and 
the Hardy-Littlewood-Sobolev inequality in the case where $V_{int}$ satisfies 
$|V_{int}(x)|\le C|x|^{-2}$, or the Young's inequality in the case where 
$V_{int}$ satisfies $V_{int}\in L^{n/2}(\R^n)$. 
In view of  the inequality $\alpha^2 \beta < 2/3 ( \alpha^3+\beta^3)$, one has 
\begin{align*}
     \| Q_H(\cdot, \g{u})u_j \|_{L^2}+ \left\| \int_{\R^n} Q_F(\cdot,y)u_j(y,t)\, dy \right\|_{L^2} 
      & \le C \left( \sum_{k=1, k\not=j}^N \|u_k\|_{L^{q}}^{3} +
      \|u_j\|_{L^{q}}^3 \right)
      =  C \| \g{u} \|^3_{L^{q}}.
\end{align*}
This completes the proof. 
\end{proof}

\begin{lem}\label{lem:Enss-Weder}
Let $n\ge 2$ and $s>1$ and put $\Phi_v (x)=e^{iv\cdot x}\varphi(x)$.  
Assume that 
$q$ is a compact operator from $L^2(\R^n)$ to $L^{2,s}(\R^n)$.  
Then for any $\varphi \in \Cal{S}_0$, 
there exist a positive constant
$C$ such that 
\[ \int_{-\infty}^{\infty} 
 \| q U_0(t) \Phi_v \|_{L^2} \, dt \le 
 \frac{C}{|v|}
\]
for $|v|$ large enough. 
\end{lem}

The proof of Lemma \ref{lem:Enss-Weder} will be found in \cite[Lemma 2.2]{Enss-Weder1995} 
and its proof.

Let $\g{u}(t)$ be the scattering solution and 
let $\Omega_-$ be the wave operator which assigns the free state $U_0(t)\g{f}^-$ 
to the interacting state $\g{u}(t)=U(t)\g{\psi}$ 
(see, e.g., Strauss \cite{Strauss1974}). 
In particular, $\Omega_- \, : \, H^1 \ni \g{f}^-\to \g{u}(0)\in H^1$.

\begin{lem}\label{lem:2-4}
Let $n= 3$ and $\g{\Phi}_v(x)=(e^{iv\cdot x}\varphi_j (x))_{1\le j\le N}$, 
$j=1,2, \cdots, N$. 
Assume that potentials $V_{int}$ and $V_{ext}$ satisfy Assumption \ref{ass:interaction-2} and Assumption \ref{ass:external}, respectively.
Then for any $\varphi_j \in \Cal{S}_0$,  
we have  
\[ \| ((\Omega_- -I)U_0(t)\g{\Phi}_v)_j \|_{L^2}=O(|v|^{-1})\]
as $|v|\to \infty$ uniformly in $t\in \R$.  
\end{lem}

The proof of Lemma \ref{lem:2-4} is quite the same as \cite[Lemma 2.3]{Watanabe2019}.

\subsection{Time-decay estimate of solutions}\label{subsec:2-2}

\begin{prop}\label{prop:2-1}
 Let $n=3$ and $a=([n/2]+1)^{-1}(n/2)=3/4$. 
 Assume that potentials $V_{int}$ and $V_{ext}$ satisfy Assumption \ref{ass:interaction-2} and Assumption \ref{ass:external}, respectively. 
 Then the solutions $u_j(t)$, $j=1,2, \cdots, N$ of \eqref{eqn:1-1} with 
 initial states $\varphi_j \in \Cal{S}_0$ satisfies
 \begin{equation*}
     \| u_j (t) \|_{L^{\infty}} \le C t^{-n/2}(\log t)^{a}, 
 \end{equation*}
 for $t\ge e$, where $C>0$. 
\end{prop}

Let us prepare for proving the Proposition \ref{prop:2-1}. 
We denote $V(tx)$ by $V^t(x)$. Putting 
 \[
   v_j(t)=(it)^{n/2} e^{-it|x|^2/2} u_j(t, tx)
 \]
gives a equation (see e.g., Wada \cite{Wada})
\begin{align}
    i\frac{\p v_j}{\p t} &= -\frac{1}{2t^2}\Delta v_j + f_j (t, \g{v}), \label{eqn:2-2-1} \\
    f_j(t, \g{v})&=V^t_{ext}(x)v_j(t) + 
    \sum_{k=1}^N \left\{ (V_{int}^t * |v_k|^2)v_j - (V_{int}^t*v_j \overline{v_k})v_k \right\}. 
    \label{eqn:2-2-2}
\end{align}
We note that the $L^2$-conservation law for $v_j(t)$ holds 
(see Isozaki \cite{Isozaki1983}): 
\begin{equation}\label{eqn:2-2-3}
    \| v_j(t) \|_{L^2} = \| u_j (t) \|_{L^2} = \| u_j (0) \|_{L^2}, 
    \qquad j=1,2, \cdots , N.
\end{equation}
Thanks to the Gagliardo-Nirenberg inequality (Lemma \ref{lem:Gagliardo})
and \eqref{eqn:2-2-3}, we have
\begin{align}
    \| u_j (t) \|_{L^{\infty}} & \le 
     t^{-n/2} \| v_j(t) \|_{L^{\infty}} \notag \\ 
     & \le C t^{-n/2} \|v_j(t)\|_{L^2}^{1-a} \| \Delta v_j(t) \|_{L^2}^{a} \notag \\
     & = C t^{-n/2} \| \Delta v_j (t) \|_{L^2}^{a} \label{eqn:2-2-4}
\end{align}
for some $C>0$. 
Therefore, estimating $\| \Delta v_j (t) \|_{L^2}$ gives the proof of 
Proposition \ref{prop:2-1}. 
In order to estimate $\| \Delta v_j (t) \|_{L^2}$, we need 

\begin{lem}\label{lem:2-2-1}
Let $n=3$. Assume that potentials $V_{int}$ and $V_{ext}$ satisfy Assumption \ref{ass:interaction-2} and Assumption \ref{ass:external}, respectively.
 Then $v_j(t)$ satisfies 
  \begin{equation*}
      \sum_{j=1}^N \| \nabla v_j (t) \|_{L^2}^2 \le C
  \end{equation*}
  for some $C>0$ and for any $t>0$.
\end{lem}

\begin{proof}
 We first note that the pseudo-conformal conservation laws for $v_j(t)$ 
 (see, e.g., Cazenave \cite[Section 7.2]{Cazenave2003}) holds: 
 \begin{equation}
     \sum_{j=1}^N \left\| \nabla v_j (t) \right\|_{L^2}^2  
      + t^2 G(t, \g{v}) = \sum_{j=1}^N \left\| x u_j(0) \right\|_{L^2}^2 
      + \int_0^t s \Theta (s, \g{v})\, ds, \label{eqn:2-2-5}
 \end{equation}
where 
\begin{align*}
     G(t, \g{v}) & = \sum_{j=1}^N \int_{\R^n} \frac{1}{2} V^t_{ext}(x) |v_j(t)|^2\, dx 
      + \sum_{j,k=1}^N \frac{1}{4} \int_{\R^n} (V^t_{int}*|v_k(t)|^2)|v_j(t)|^2\,dx  \\
    & \hspace{1em}
      - \sum_{j,k=1}^N \frac{1}{4} \int_{\R^n} (V^t_{int}*v_j(t)\overline{v_k}(t))v_k(t) \overline{v_j}(t)\,dx,  \\ 
     \Theta(t, \g{v}) &=  \sum_{j=1}^N \int_{\R^n} \left( V^t_{ext}(x) +\frac{1}{2} x\cdot (\nabla V^t_{ext})(x) \right) |v_j(t)|^2\, dx \\
    & \hspace{1em}  + \sum_{j,k=1}^N  \int_{\R^n} \left( \left(V^t_{int}+\frac{1}{2}x\cdot \nabla V^t_{int} \right)*|v_k(t)|^2\right)|v_j(t)|^2\,dx  \\
    & \hspace{1em}
      - \sum_{j,k=1}^N  \int_{\R^n} \left( \left(V^t_{int}+\frac{1}{2}x\cdot \nabla V^t_{int} \right)*v_j(t)\overline{v_k}(t)\right)v_k(t) \overline{v_j}(t)\,dx.
 \end{align*}
In view of the Assumption \ref{ass:external}, Assumption \ref{ass:interaction-2} and 
the Cauchy-Schwarz inequality, we find that 
$G(t, \g{v})\ge 0$ and $\Theta (t, \g{v}) \le 0$ for $t \ge 0$. 
It therefore follows from \eqref{eqn:2-2-5} that 
\begin{equation*}
    \frac{d}{dt}  \left( \sum_{j=1}^N \| \nabla v_j(t) \|_{L^2}^2 + t^2 G(t, \g{v}) \right) 
    = t \Theta (t, \g{v}) \le 0
\end{equation*}
for $t\ge 0$, which implies that $\sum_{j=1}^N \| \nabla v_j(t) \|_{L^2}^2  \le C$.
\end{proof}

{\it Proof of Proposition \ref{prop:2-1}.} 
We are now in a position to prove Proposition \ref{prop:2-1}. 
In the proof, we abbreviate the $L^p$-norm of a function $f$ as $\| f \|_p$. 
Applying $\Delta$ to the equation \eqref{eqn:2-2-1}, one has 
\begin{equation}\label{eqn:2-2-6}
 i\frac{\p}{\p t}\Delta v_j(t) 
 = -\frac{1}{2t^2}\Delta^2 v_j(t) + \Delta f_j (t, \g{v}).
\end{equation}
Multiplying \eqref{eqn:2-2-6} by $\Delta \overline{v_j}$ and 
integrating the imaginary part over $\R^n$, we have 
\[
   \frac{1}{2} \frac{d}{dt} \left\| \Delta v_j (t) \right\|_{L^2}^2 = {\rm Im} \int_{\R^n} \Delta f_j(t,\g{v}) 
   \Delta \overline{v_j} (t) \, dx.
\]
Due to the fact that integrals 
\begin{align*}
    &{\rm Im} \int_{\R^n} V^t_{ext}(x) | \Delta v_j(t) |^2 \, dx, \quad 
    {\rm Im} \sum_{j,k=1}^N\int_{\R^n} (V^t_{int}*|v_k(t)|^2) | \Delta v_j(t) |^2 \, dx, \\
    &{\rm Im} \sum_{j,k=1}^N\int_{\R^n} 
     (V^t_{int}*v_j(t) \overline{v_k}(t)) \Delta v_k(t) 
     \Delta \overline{v_j}(t) \, dx
\end{align*}
vanish, one gets
\begin{equation*}
    \frac{1}{2} \sum_{j=1}^N \frac{d}{dt} 
    \left\| \Delta v_j (t) \right\|_{L^2}^2 = 
    \sum_{\ell=1}^6 I_{\ell}(t),
\end{equation*}
where 
\begin{align*}
    I_1(t) &= 
     {\rm Im} \sum_{j=1}^N 
     \int_{\R^n} (\Delta V^t_{ext})(x) v_j(t) \Delta \overline{v_j}(t) \, dx, \\
    I_2(t) &= 2
     {\rm Im} \sum_{j=1}^N 
     \int_{\R^n} (\nabla V^t_{ext})(x)\cdot  \nabla v_j(t) 
     \Delta \overline{v_j}(t) \, dx, \\
    I_3(t) &= 
     {\rm Im} \sum_{j,k=1}^N 
     \int_{\R^n} v_j(t) \Delta \overline{v_j}(t) 
     \left( \Delta ( V^t_{int}*|v_k(t)|^2)\right)(x) \, dx, \\
    I_4(t) &= 2 
     {\rm Im} \sum_{j,k=1}^N 
     \int_{\R^n} \Delta \overline{v_j}(t) 
     \left( \nabla v_j(t)\cdot  
     \nabla ( V^t_{int}*|v_k(t)|^2)\right)(x) \, dx, \\
    I_5(t) &= - 
     {\rm Im} \sum_{j,k=1}^N 
     \int_{\R^n} v_k(t) \Delta \overline{v_j}(t) 
     \left( \Delta  
      ( V^t_{int}*v_j(t) \overline{v_k}(t))\right)(x) \, dx, \\
    I_6(t) &= -2
     {\rm Im} \sum_{j,k=1}^N 
     \int_{\R^n} \Delta \overline{v_j}(t)  
     \left( \nabla v_k (t)\cdot  
     \nabla ( V^t_{int}*v_j(t)\overline{v_k}(t)) \right)(x) \, dx.
\end{align*}
We shall show that $|I_1|, |I_2| \le C t^{-\gamma} \sum_{j=1}^N \| \Delta v_j(t) \|_{L^2}$. 
Due to the Assumption \ref{ass:external} that 
$V_{ext}(x)$ is a homogeneous function of degree $-\gamma$, 
one has
\[
   (\nabla V^t_{ext})(x)=t^{-\gamma} (\nabla V_{ext})(x), 
   \qquad 
   (\Delta V^t_{ext}) (x) = t^{-\gamma} (\Delta V_{ext})(x).
\]
By using the Schwartz inequality, H\"{o}lder inequality, Gagliardo-Nirenberg inequality (Lemma \ref{lem:Gagliardo}) and 
Lemma \ref{lem:2-2-1}, we obtain
\begin{align*}
    | I_1 (t) | & \le 
     t^{-\gamma} \sum_{j=1}^N \| (\Delta V_{ext}) v_j(t)\|_{2} 
     \| \Delta v_j (t) \|_{2} \\
    & \le 
     t^{-\gamma} \| \Delta V_{ext} \|_n \sum_{j=1}^N 
     \| v_j(t)\|_{\frac{2n}{n-2}} \| \Delta v_j (t) \|_2 \\
    & \le C
     t^{-\gamma} \| \Delta V_{ext} \|_n \sum_{j=1}^N 
     \| \nabla v_j(t)\|_{2} \| \Delta v_j (t) \|_2 \\
    & \le C
     t^{-\gamma}  \sum_{j=1}^N 
      \| \Delta v_j (t) \|_2
\end{align*}
and 
\begin{align*}
    | I_2 (t) | & \le 
     t^{-\gamma} \| \nabla V_{ext} \|_{\infty} 
     \sum_{j=1}^N \|  \nabla v_j (t) \|_2 
     \| \Delta v_j (t) \|_{2} \\
    & \le C
     t^{-\gamma}  \sum_{j=1}^N 
      \| \Delta v_j (t) \|_2.
\end{align*}
We next claim that 
$| I_{\ell} (t) | \le C t^{-1} \sum_{j=1}^N \| \Delta v_j (t)\|_2$, $\ell = 3,4,5,6$. 
It is easy to check that 
\begin{align*}
    \int_{\R^n} v_j (t) \Delta \overline{v_j}(t) 
    \left( \Delta ( V^t_{int}*|v_k(t)|^2 )\right)(x)\, dx 
    &= 
    \int_{\R^n} v_j (t) \Delta \overline{v_j}(t) 
    \sum_{m=1}^3 ( \p_{x_m}V^t_{int}* \overline{v_k}(t) 
    \p_{x_m}v_k(t))(x)\, dx \\
    & \hspace{1em} + 
    \int_{\R^n} v_j (t) \Delta \overline{v_j}(t) 
    \sum_{m=1}^3 ( \p_{x_m}V^t_{int}* v_k(t) 
    \p_{x_m}\overline{v_k}(t))(x)\, dx.
\end{align*}
By using the Schwartz inequality, H\"{o}lder inequality, 
Young inequality and Gagliardo-Nirenberg inequality (Lemma \ref{lem:Gagliardo}), we have 
\begin{align*}
   &\left| 
     \int_{\R^n} v_j (t) \Delta \overline{v_j}(t) 
    \sum_{m=1}^3 ( \p_{x_m}V^t_{int}* \overline{v_k}(t) 
    \p_{x_m}v_k(t))(x)\, dx
    \right| 
    \le 
    \sum_{m=1}^3 \|\left\{ \p_{x_m}V^t_{int}*\overline{v_k}(t) 
    \p_{x_m}v_k (t)\right\}v_j(t) \|_2 \| \Delta \overline{v_j} (t) \|_2 \\
   & \hspace{3em} 
    \le 
    \| \Delta \overline{v_j}(t) \|_2 \| v_j(t) \|_{\frac{2n}{n-2}}
    \sum_{m=1}^3 \| \p_{x_m}V^t_{int}*\overline{v_k}(t) 
    \p_{x_m}v_k (t) \|_n \\
   & \hspace{3em} 
    \le C
    \| \Delta \overline{v_j}(t) \|_2 \| v_j(t) \|_{\frac{2n}{n-2}}
    \sum_{m=1}^3 \| \p_{x_m}V^t_{int} \|_{\frac{n}{2}} \| \overline{v_k}(t) 
    \p_{x_m}v_k (t) \|_{\frac{n}{n-1}} \\
   & \hspace{3em} 
    \le C
    \| \Delta \overline{v_j}(t) \|_2 \| v_j(t) \|_{\frac{2n}{n-2}}
    \sum_{m=1}^3 \| \p_{x_m}V^t_{int} \|_{\frac{n}{2}} \| \overline{v_k}(t) \|_{\frac{2n}{n-2}}
    \| \p_{x_m}v_k (t) \|_{2} \\
  & \hspace{3em} 
    \le C t^{-1}
    \| \Delta \overline{v_j}(t) \|_2 
    \| \nabla v_j (t) \|_{2} \| \nabla \overline{v_j}(t) \|_2
    \sum_{m=1}^3 \| \p_{x_m}V_{int} \|_{\frac{n}{2}}  
    \| \p_{x_m}v_k (t) \|_2 \\
  & \hspace{3em}
    \le C t^{-1} \| \Delta v_j (t) \|_2,
\end{align*}
which implies that 
$|I_3 (t) |, |I_5 (t) | \le C t^{-1}\sum_{j=1}^N \| \Delta v_j (t) \|_2$. 
Here we have used the fact that 
$\|V^t\|_p = t^{-n/p}\|V\|_p$. 

Similarly, one has
\begin{align*}
   &\left| 
     \int_{\R^n} \Delta \overline{v_j}(t) \nabla v_j (t) \cdot  
    \nabla  \left( ( V^t_{int}* |v_k(t)|^2 )
    \right)(x) \, dx
    \right| 
    \le 
     \int_{\R^n} \left| 
      \Delta \overline{v_j}(t) \nabla v_j \cdot 
      (\nabla V^t_{int}*v_k (t) \overline{v_k}(t))(x) \right| 
      \, dx
     \\
   & \hspace{3em} 
    \le 
    \| \Delta \overline{v_j}(t) \|_2 
    \| \nabla v_j(t) \|_{2}
    \| \nabla V^t_{int}* v_k (t) \overline{v_k}(t) \|_{\infty} \\
   & \hspace{3em} 
    \le C
    \| \Delta \overline{v_j}(t) \|_2 
    \| \nabla v_j(t) \|_{2}
    \| \nabla V^t_{int}\|_{\frac{n}{2}} 
    \| v_k (t) \overline{v_k}(t) \|_{\frac{n}{n-2}} \\
   & \hspace{3em} 
    \le t^{-1}
    \| \Delta \overline{v_j}(t) \|_2 
    \| \nabla v_j(t) \|_{2}
    \| \nabla V_{int} \|_{\frac{n}{2}} 
    \| v_k (t)  \|_{\frac{2}{n-2}}^2 \\
  & \hspace{3em} 
    \le t^{-1}
    \| \Delta \overline{v_j}(t) \|_2 
    \| \nabla v_j(t) \|_{2}
    \| \nabla V_{int} \|_{\frac{n}{2}} 
    \| \nabla v_k (t)  \|_{2}^2, \\
\end{align*}
which implies that 
$|I_4 (t) |, |I_6 (t) | \le C t^{-1}\sum_{j=1}^N \| \Delta v_j (t) \|_2$.

In view of the inequalities for $|I_{\ell}(t) |$, $\ell=1,2, \cdots 6$, 
we obtain a differential inequality:
\begin{equation*}
    \frac{1}{2} \sum_{j=1}^N \frac{d}{dt} \| \Delta v_j (t) \|_2^2 \le 
    C(t^{-\gamma}+t^{-1}) \left(
    \sum_{j=1}^N \| \Delta v_j (t) \|_2^2 \right)^{1/2}
\end{equation*}
for $t>0$, where $C>0$ and $\gamma \ge 1$. 
This differential inequality implies that 
\begin{equation*}
    \sum_{j=1}^N \| \Delta v_j(t) \|_2 \le C ( \log t +1) 
    \le C \log t 
\end{equation*}
for $t\ge e$. 
Hence, from \eqref{eqn:2-2-4}, we obtain the desired estimate. 
\par
\hfill $\Box$

\subsection{Proof of Theorem \ref{thm:2-1}}\label{subsec:2-3}
Let $u_j \in \Cal{W}$ be a scattering solution to \eqref{eqn:1-1} 
with initial states $e^{-itH_0}\varphi_j$ at $t=-\infty$. 
Because the scattering solution satisfies the integral equation 
\begin{equation*}
    u_j(t)=e^{-itH_0} \varphi_{j}+\frac{1}{i} \int_{-\infty}^{t} e^{-i(t-\tau)H_0}P_{j}(x, \g{u} ), d\tau, 
\end{equation*}
we obtain 
\begin{align*}
    \left( e^{-itH_0}\psi_j\right)(x) 
    &= \left( e^{-itH_0}\varphi_j\right)(x) + \frac{1}{i} 
    \int_{\R} e^{-i(t-\tau)H_0}P_j(x,\g{u})\, d\tau \\
    &= u_j(t,x) + \frac{1}{i} \int_{t}^{\infty}e^{-i(t-\tau)H_0}P_j(x,\g{u})\, d\tau .
\end{align*}
Thanks to Proposition \ref{prop:2-1} and Lemma \ref{lem:2-2}, for $t\ge e$, 
one has 
\begin{align*}
   \| u_j(t) - e^{-itH_0}\psi_j \|_{L^2} 
   & \le 
   \int_t^{\infty} \| e^{-i(t-\tau)H_0} V_{ext}u_j(\tau) \|_{L^2} \, d\tau \\
   & \hspace{1em} +
   \int_t^{\infty} 
    \left\| 
      e^{-i(t-\tau)H_0}\left( Q_H(\cdot, u)u_j(\tau) + \int_{\R^n} Q_F(x,y,\g{u})u_j(y,\tau)\, dy \right)
    \right\|_{L^2} \, d\tau \\
    & \le \int_t^{\infty} \| V_{ext} u_j(\tau) \|_{L^2}\, d\tau +
     C\int_{t}^{\infty}\| u_j(\tau) \|_{L^{q}}^3\, d\tau \\
    & \le C_1 \| V_{ext} \|_{L^2} \int_t^{\infty} \tau^{-3/2}(\log \tau )^{3/4} \, d\tau 
     + C_2 \int_{t}^{\infty} \| u_j( \tau ) \|_{L^{q}}^3 \, d\tau \\
    & \longrightarrow 0 \qquad \text{as $t\to \infty$}
\end{align*}
for some $C_1$, $C_2>0$, 
due to the fact that $u_j \in L^3(\R ; L^{q})$ and 
\begin{equation*}
    \int_{t}^{\infty} \tau^{-3/2}(\log \tau )^{3/4}\, d\tau 
    < 2^{7/4}\int_0^{\infty}e^{-\mu} \mu^{3/4}\, d\mu 
    =2^{7/4} \Gamma\left( \frac{7}{4} \right),
\end{equation*}
where $\Gamma(s)$ is the Gamma function.
This completes the proof. 
\par
\hfill $\Box$

\section{Asymptotics of the scattering operator}
\label{sec:3}
Let $\g{\Phi}_v (x) = e^{iv\cdot x}\g{\varphi} (x)$. 
The components of the vector $\g{\Phi}_v$ is denoted as 
$(\g{\Phi}_v)_j$ $(j=1,2,\cdots, N)$.
Put $I_j(v)=i< ((S-I)\g{\Phi}_v)_j, (\g{\Phi}_v)_j>_{L^2}$. 
We consider the asymptotic behavior of the function $I_j(v)$ as $|v|\to \infty$.
The $X$-ray transform of a function $f$ is defined to be 
\[
   (Xf)(x, \theta)=\widetilde{f}(x,\theta)=\int_{-\infty}^{\infty} f(x+\theta t) \, dt,
\]
where $x\in \R^n$ and $\theta \in \bs^{n-1}$.

\begin{thm}\label{thm:3-1}
Let $ n=3$. 
Assume that potentials $V_{int}$ and $V_{ext}$ satisfy Assumption \ref{ass:interaction-2} and Assumption \ref{ass:external}, respectively. 
Then for $|v|$ sufficiently large and for any $\varphi_j\in \Cal{S}_0$, 
$j=1,2,\cdots, N$, 
the function $I_j(v)$ has the expansion 
\begin{equation*}
     I_j(v) =
      \int_{\R^n} \widehat{V_{int}}(\xi) H_j(\xi) \, d\xi 
       + 
      \frac{1}{|v|} \left< 
       \widetilde{V_{ext}}(\cdot, \widehat{v}) 
      \varphi_j, \varphi_j\right>_{L^2} + 
     O(|v|^{-2})
\end{equation*}
as $|v|\to \infty$, where $\widehat{v}=v/|v|\in \bs^{n-1}$ and 
\begin{align*}
   H_j(\xi) &= \sum_{k=1}^N 
                \int_{\R}  \Cal{F} \left( |U_0 (t) \varphi_k \right|^2 ) (\xi) 
                \overline{\Cal{F} \left( |U_0 (t) \varphi_j \right|^2 ) (\xi)} \, dt \\
             &\hspace{2em}- \sum_{k=1}^N \int_{\R}  \left| 
                 \Cal{F} \left(  \left(U_0(t) \varphi_j \right) 
                \overline{ \left( U_0(t) \varphi_k \right) } \right) (\xi) \right|^2
                 \, dt. 
\end{align*}
\end{thm}

\begin{proof}
In view of the representation of the scattering operator \eqref{eqn:1-2}, we break the function $I_j(v)$ in two parts: 
\begin{equation*}
    I_j(v) = I_j^{(0)}(v) +I^{(1)}_j(v),
\end{equation*}
where 
\begin{align*}
    I_j^{(0)}(v) &= \int_{\R} 
    \left< N_j(\cdot,\g{u}), U_0(s) (\g{\Phi}_v)_j \right>_{L^2}\, ds, \\
    I_j^{(1)}(v) &= \int_{\R} 
    \left< V_{ext} u_j(s), U_0(s) (\g{\Phi}_v)_j \right>_{L^2}\, ds.
\end{align*}
We know (see \cite[subsection 2.2]{Watanabe2019}) that the function 
$I^{(0)}_j(v)$ can be expanded as 
\begin{equation*}
    I_j^{(0)}(v)=\int_{\R^n} \widehat{V_{int}} H_j(\xi)\, d\xi + R_1(v)
\end{equation*}
with the estimate $|R_1(v)|\le C |v|^{-2}$ for some $C>0$ and 
for $|v|$ sufficient large.

We will claim that 
\begin{equation*}
    I_j^{(1)}(v) = \frac{1}{|v|} 
    \left<
     \widetilde{V_{ext}}(\cdot, \widehat{v})\,  
     \varphi_j, \varphi_j
    \right>_{L^2} + O(|v|^{-2})
\end{equation*}
as $|v|\to \infty$. 
Let $\Omega_-$ be a wave operator. Then we have 
\begin{align}
    I_j^{(1)}(v) &= 
     \left<
      \int_{\R} U_0(-t) V_{ext} \left\{ u_j(t)-U_0(t)(\g{\Phi}_v)_j\right\}
      \, dt, (\g{\Phi}_v)_j
     \right>_{L^2} 
      +
     \left<\int_{\R} U_0(-t) V_{ext} U_0(t) (\g{\Phi}_v)_j\, dt, (\g{\Phi}_v)_j
     \right>_{L^2} \notag \\ 
     &=
     \left<
      \int_{\R} [(\Omega_- -I)(U_0(t)\g{\Phi}_v)]_j\, dt, V_{ext}U_0(t)(\g{\Phi}_v)_j
     \right>_{L^2} \notag \\
     & \hspace{1em} +
      \frac{1}{|v|} 
      \left<
        \widetilde{V_{ext}}(\cdot ,\widehat{v}) U_0(\tau/v) \varphi_j, \,  
    U_0(\tau /v)\varphi_j
     \right>_{L^2} \notag \\
     & \le 
      \| [( \Omega_- - I)U_0(t)\g{\Phi}_v]_j\|_{L^2} 
      \left\| \int_{\R}V_{ext}U_0(t)(\g{\Phi}_v)_j\, dt 
      \right\|_{L^2} \notag \\
     & \hspace{1em} + 
     \frac{1}{|v|} 
     \left< 
      \widetilde{V_{ext}} (\cdot , \widehat{v}) 
       U_0(\tau /v)\varphi_j, \,  
       U_0(\tau/v)\varphi_j
     \right>_{L^2}. \label{eqn:3-1}
\end{align}
Thanks to Lemma \ref{lem:Enss-Weder} and Lemma \ref{lem:2-4}, 
the first term in \eqref{eqn:3-1} is estimated as $O(|v|^{-2})$ 
for $|v|$ sufficiently large. We know (see \cite{Enss-Weder1995}) that the second term in \eqref{eqn:3-1} is equal to 
\begin{equation*}
    \frac{1}{|v|} \left<
    \widetilde{V_{ext}} (\cdot , \widehat{v}) \varphi_j, \, \varphi_j
    \right>_{L^2} + O(|v|^{-2})
\end{equation*}
for $|v|$ sufficiently large. 
The proof is completed. 
\end{proof}

\section{Reconstructions}
\label{sec:4}
We complete the proof of Theorem \ref{thm:1-1} and 
give reconstruction formulas.

\subsection{Reconstruction of the interaction 
potential}
Let $\Gamma \subset \R$ be a the compact set and 
$\g{\Phi}_{v}(x, \lambda)= e^{iv\cdot x}\g{\varphi}((\lambda+1)x)$. 
Put 
\begin{equation*}
    S_{j}^{lim}(\lambda) = \lim_{|v|\to \infty} i 
    \left< 
     ((S-I)\g{\Phi}_v(\cdot, \lambda))_j, 
     (\g{\Phi}_v)_j(\cdot, \lambda)
    \right>_{L^2}, \qquad j=1,2,\cdots, N.
\end{equation*}
In view of Theorem \ref{thm:3-1}, we have 
\begin{equation}\label{eqn:4-1}
    S^{lim}_{j}(\lambda)=\int_{\R^n} \widehat{V_{int}}(\xi) 
    H_j(\xi, \lambda)\, d\xi
\end{equation}
for any $\varphi_j\in \Cal{S}_0$. 
Due to the fact (see \cite[Theorem 1.12]{Watanabe2019}) that the equation \eqref{eqn:4-1} is an integral equation of the first kind with a compact operator from $H^k(\R^n)$ to $L^2(\Gamma)$ for  $k>n/2$, we can reconstruct $\widehat{V_{int}}$ from the scattering  operator by using the theory of integral equations (see e.g., \cite[Section 15.4]{Kress}) or approximate techniques (see e.g., \cite[Section 8.3]{Morse-Feshbach}). 
For example, the Picard's theorem allows us to obtain a reconstruction formula of $\widehat{V_{int}}$.  

\begin{defn}
Let $X$ and $Y$ be  Hilbert space, $\Cal{A}\, :\, X\to Y$ be a compact linear operator, and 
$\Cal{A}^* \, : \, Y \to X$ be its adjoint.  
Singular values of $\Cal{A}$ is 
the non-negative square roots of the eigenvalue of non-negative self-adjoint compact 
operator $\Cal{A}^* \Cal{A} \, : \, X \to X$. 
The singular system of $\Cal{A}$ is 
the system $\{ \mu_n, \varphi_n, g_n\}$, $n\in \mathbb{N}$,
where 
 $\varphi_n \in X$ and $g_n\in Y$ are orthonormal sequences such that 
$\Cal{A}\phi_n = \mu_n g_n$ and $\Cal{A}^*g_n = \mu_n \phi_n$ for all 
$n\in \mathbb{N}$. 

\end{defn}

We denote the null-space of the operator $T$ by $\Cal{N}(T)$. 
 \begin{thm}
  Let $n=3$. 
  Assume that potentials $V_{int}$ and $V_{ext}$ satisfy Assumption \ref{ass:interaction-2} and Assumption \ref{ass:external}, respectively.     
Then for any  $\varphi_j \in  \Cal{S}_0$, $j=1,2, \cdots, N$ the function 
$S^{lim}_j (\lambda)$
  is the $L^2$-function on a compact set $\Gamma \subset \R$. 
  Moreover, 
  letting $\{ \mu_n, \phi_n, g_n\}$, $n\in \mathbb{N}$ be a singular system of 
  the integral operator $T$: 
  \begin{equation*}
      (Tf)(\lambda):= \int_{\R^n} f(\xi) H_j(\xi, \lambda)\, d\xi,
  \end{equation*}
  the Fourier transform of the interaction potential is reconstructed by the formula:      
    \[
       \widehat{V_{int}}(\xi)= \sum_{n=1}^{\infty} \frac{1}{\mu_n} 
       \left< S^{lim}_j , g_n \right>_{L^2(\Gamma)} \phi_n
    \]
   if and only if  $S^{lim}_j\in \Cal{N}(T^*)^{\perp}$ and satisfies
      \[
       \sum_{n=1}^{\infty} \frac{1}{\mu_n^2} \left|
         \left<S^{lim}_j , g_n \right>_{L^2(\Gamma)} \right|^2 <\infty.
    \]
 \end{thm}

\begin{rem}
 The procedure of constructing the singular system is as follows: 
 Due to the fact that the operator  $T^*T$ is a self-adjoint compact operator on $H^k(\R^n)$ for $k >n/2$, the operator $T^*T$ has at least one eigenvalues different from zero and at most a countable set of eigenvalues accumulating only at zero. 
 Let $(\phi_n)$ denote an orthonormal sequences such that 
 $T^*T\phi_n = \mu_n^2 \phi_n$. Then we can define  $g_n$ as 
 $g_n = \mu_n^{-1} T^{*}\phi_n$. 
\end{rem}

Uniqueness of identifying $\widehat{V_{int}}$ follows from \cite[Theorem 1.14]{Watanabe2019}. Thus we conclude that one can uniquely determine $\widehat{V_{int}}$ from $S$.

\subsection{Reconstruction of the external 
potential}
Assume that potentials $V_{int}$ and $V_{ext}$ satisfy Assumption \ref{ass:interaction-2} and Assumption \ref{ass:external}. 
In addition, suppose that 
\begin{equation}\label{eqn:4-2-1}
    \int_{K}^{\infty} (1+R) \| V_{ext}(x)F(|x|\ge R)\|_{L^{\infty}(\R^n)} \, dR <\infty, \qquad K>0, 
\end{equation}
where $F(A)$ is the characteristic function of $A\subset \R^n$.
We note that the Assumption \ref{ass:external-2} includes the 
condition \eqref{eqn:4-2-1}. Hence, the following also holds 
under the  Assumption \ref{ass:interaction-2} and 
Assumption \ref{ass:external-2}.
In view of Theorem \ref{thm:3-1} with the help of 
Takiguchi \cite[Proposition 3.2.]{Takiguchi1998}, 
it is easy to verify that  
\begin{equation*}
    \lim_{|v|\to \infty} |v| \left< i
    \left( e^{-iv\cdot x}(S-I)\g{\Phi}_{v}\right)_j - 
  \int_{\R} U_0(-t) N_{j}(x, U_0(t)\g{\varphi})\, dt, \psi 
    \right>_{L^2}
    =\left< 
     \widetilde{V_{ext}}(\cdot, \widehat{v}) \varphi_j, \psi
    \right>_{L^2}
\end{equation*}
for any $\varphi_j \in \Cal{S}_0$ and 
for any $\psi \in \Cal{S}$.  
From Theorem \ref{thm:4-2}, $\widehat{V_{int}}(\xi)$ is the known function. 
Therefore, $N_{j}(x, U_0(t)\g{\varphi})$ is known function. 
Hence, the above identity shows that one can determine the $X$-ray transform of $V_{ext}$ from $S$ in the sense of the tempered distribution $\Cal{S}'$. 
By using the inversion formula for the $X$-ray transform, we obtain a reconstruction formula of $V_{ext}$. 
More precisely, define the operator $I^a$, which is called the Riesz potential, 
as 
\[ I^{a}f := \Cal{F}^{-1} ( |\xi|^{-a}\widehat{f}(\xi)), \qquad a<n.
\]
We denote a hyperplane passing  through the origin and orthogonal to 
$\alpha \in \bs^{n-1}$ by $\alpha^{\perp}$. 
Let $I^a_{\alpha^{\perp}}$ be the $(n-1)$-dimensional Riesz potential acting 
on the hyperplane $\alpha^{\perp}$. 
The adjoint of the X-ray transform is denoted as $X^*$: 
\[ (X^{*}g)(x)=\int_{\bs^{n-1}}g(\theta, x-(\theta\cdot x)\theta)\, d\sigma, 
\]
where $d\sigma$ is the Lebesgue measure on the unit sphere $\bs^{n-1}$ in 
$\R^n$. 
We know (see, e.g., Ramm-Katsevich \cite[Theorem 2.6.2.]{Ramm-Katsevich1996}) 
that 
letting $f\in \Cal{S}(\R^n)$ and $g=Xf$, one has for any $|a|<n$ 
\begin{equation*}
    f=\frac{1}{2\pi |\bs^{n-2}|}I^{-a}X^* I_{\alpha^{\perp}}^{a-1}g,
\end{equation*}
where $|\bs^{n-2}|$ is the surface area of $\bs^{n-2}$.
This inversion formula also holds for the tempered distribution 
$f\in \Cal{S}'(\R^n)$ and $g=Xf \in \Cal{S}'(T)$, where 
$T=\alpha^{\perp} \times \bs^{n-1}$ (see Takiguchi \cite{Takiguchi1998}). 
Thus, we obtain 
\begin{thm}\label{thm:4-2}
Let $ n =3 $ and $|a|<n $. 
Assume that potentials $V_{int}$ and $V_{ext}$ satisfy Assumption \ref{ass:interaction-2} and Assumption \ref{ass:external} with 
\eqref{eqn:4-2-1}, respectively. 
Then for any $\varphi_j \in \Cal{S}_0$, we have 
\[ V_{ext}=\frac{I^{-a} X^* I^{a-1}_{\alpha^{\perp}}}{2\pi |\bs^{n-2}|}
\frac{1}{\varphi_j}
 \lim_{|v|\to \infty}|v|
 \left\{ i
 \left( e^{-iv\cdot x}(S-I)\g{\Phi}_{v}\right)_j - 
  \int_{\R} U_0(-t) N_{j}(x, U_0(t)\g{\varphi})\, dt
  \right\}
\]
in $\Cal{S}'$.  
\end{thm}

\section*{Acknowledgement} 
This work was supported by JSPS KAKENHI Grant Number 19K03617. 

\bibliographystyle{unsrt}  
\bibliography{michiyukirefs2}  

\end{document}